\documentclass{article}

\usepackage{lineno,hyperref}
\modulolinenumbers[5]
\usepackage{latexsym,graphicx}
\usepackage{amsmath}
\usepackage{amssymb}
\usepackage{float}
\usepackage{url}
\usepackage[marginparwidth=2cm]{geometry}
\usepackage{amsthm}
\usepackage{tikz}
\usetikzlibrary{calc}
\usepackage{relsize}
\usepackage{cleveref}
\tikzset{fontscale/.style = {font=\relsize{#1}}}
\usepackage[noend]{algorithmic}
\usepackage{algorithm}
\usepackage{dsfont}
\usetikzlibrary{positioning}
\usetikzlibrary{decorations.pathreplacing}

\newtheorem{theorem}{Theorem}
\newtheorem{proposition}{Proposition}
\newtheorem{example}{Example}
\newtheorem{corollary}{Corollary}
\newtheorem{remark}{Remark}
\newtheorem{lemma}{Lemma}

\begin{document}
{\large
\bigskip
\vspace{7 mm}

\centerline{\large \bf Firefighting on Trees}

\vspace{5 mm}

\centerline{Pierre Coupechoux$^{\,a,}$, Marc Demange$^{\,b,}$, David Ellison$^{\,b,}$, Bertrand Jouve$^{\,c,}$}
\vspace{3 mm}
}
%{\small
\centerline{\small a. LAAS-CNRS, Universit\'e de Toulouse, CNRS, Toulouse, France,}
\centerline{\small b. RMIT University, School of Science, Melbourne, Australia}
\centerline{\small c. LISST, CNRS, Universit\'e de Toulouse, Toulouse, France}
%}
\vspace{6 mm}
\centerline{\small pierre.coupechoux@laas.fr, marc.demange@rmit.edu.au, david.ellison2@rmit.edu.au, bertrand.jouve@cnrs.fr}
%%%%%%%%%%%%
\vspace{6 mm}
\begin{abstract}
In the {\sc Firefighter} problem, introduced by Hartnell in 1995, a fire spreads through a graph while a player chooses which vertices to protect in order to contain it. In this paper, we 
focus on the case of trees and we consider as well the {\sc Fractional Firefighter} game where the amount of protection allocated to a vertex lies between 0 and 1. While most of the work in this area deals with a constant amount of firefighters available at each turn, we  consider three research questions which arise when including the sequence of firefighters as part of the instance. We first introduce the online version of both {\sc Firefighter} and {\sc Fractional Firefighter}, in which the number of firefighters available at each turn is revealed over time. We show that  
a greedy algorithm on finite trees 
 is 1/2-competitive for both online versions, which generalises a result previously known for special cases of {\sc Firefighter}. 
We also show that the optimal competitive ratio of online {\sc Firefighter}  ranges between 1/2 and the inverse of the golden ratio. 
Next, given two firefighter sequences, we discuss sufficient conditions for the  existence of an infinite tree that {\em separates} them, in the sense that the fire can be contained with one sequence but not with the other. To this aim, we study a new purely numerical game called {\em targeting game}. Finally, we give sufficient conditions for the fire to be contained, expressed as the asymptotic comparison of the number of firefighters and the size of the tree levels.
\end{abstract}

%\begin{keyword}
%\texttt{elsarticle.cls}\sep \LaTeX\sep Elsevier \sep template
%\MSC[2010] 00-01\sep  99-00
%\end{keyword}

\section{Introduction and Definitions}\label{sec:intro} 
\subsection{Context}

Since it was formally introduced by B. Hartnell in 1995~(\cite{hartnell}, cited in~\cite{Finbow.MacGillivray}) the firefighting problem \-- {\sc Firefighter} \-- has raised the interest of many researchers.  While this game started as a very simple model for fire spread and containment problems for wildfires, it can also represent any kind of threat able to spread sequentially in a network (diseases, viruses, rumours, flood \ldots).

It is a deterministic discrete-time one-player game defined on a graph. In the beginning, a fire breaks out on a vertex and at each step, if not blocked, the fire spreads to all adjacent vertices. In order to contain the fire, the player is given a number $f_i$ of firefighters at each turn $i$ and can use them to protect vertices which are neither burning nor already protected. The game terminates when the fire cannot spread any further. In the case of finite graphs  
the aim is to save as many vertices as possible, while in the infinite case, the player wins if the game finishes, which means that the fire is contained.

This problem and its variants give rise to a generous literature; the reader is referred to~\cite{Finbow.MacGillivray} for a broad presentation of the main research directions. A significant amount of theoretical work deals with its complexity and approximability behaviour in various classes of graphs~\cite{bazgan-chopin-ries,Cai.Verbin.Yang,Finbow.king,Fomin.fedor} and its parametrised complexity~(e.g. \cite{bazgan.chopin.fellow,Cai.Verbin.Yang}). In particular, when one firefighter is available at each turn it is known to be polynomially solvable in some classes  of graphs, which include graphs of maximum degree~3 if the fire breaks out on a vertex of degree at most~2~\cite{Finbow.king}, interval graphs, permutation graphs and split graphs~\cite{Fomin.fedor}.
However it is known to be very hard, even in some restrictive cases. In particular, the case of trees was revealed to be very rich and a lot of research focuses on it.
The problem, with the same number of firefighters at each turn, is NP-hard on finite trees of maximum degree~3~\cite{Finbow.king}, as well as in even more restricted cases~\cite{bazgan-chopin-ries}; the reader is also referred to~\cite{Chlebikova-tcs} for further complexity results.
Regarding approximation results on trees, a greedy strategy was first shown to be a $\frac 12$-approximation algorithm ~\cite{Hartnell.Li} if a fixed number of firefighters is available at each turn. For a single firefighter, a $(1-\frac 1e)$-approximation algorithm is proposed in ~\cite{Cai.Verbin.Yang} for the problem in trees. This ratio was improved in~\cite{improved-treedeg3} for ternary trees and, very recently, a polynomial time approximation scheme was obtained in trees~\cite{ptas-tree}; which essentially closes the question of approximating the firefighter problem in trees with one firefighter and motivates considering some generalisations. The problem is hard to approximate within $n^{1-\varepsilon}$ on general graphs and with a single firefighter~\cite{Approx-algorithmica}.

Most papers on this subject deal with a constant firefighter sequence. In fact, the problem was originally defined with one firefighter per turn. The case of infinite grids is of particular interest and has led to the model being extended by varying the available resources per turn. The change was motivated by the fact that a fire of any size on a 2-dimensional infinite grid can be contained with two firefighters per turn but not with one~\cite{fogarty,wang-moeller}. In order to refine these results, M.-E. Messinger started considering periodic firefighter sequences~\cite{Messinger08} while more general sequences are considered in~\cite{3/2firefighter}.
A related research direction investigates integer linear programming models for the problem, especially on trees~\cite{ptas-tree,Hartke.Stephen,MacGillivray.wang}. This line of research makes very natural a {\em relaxed} version  where the amount of firefighters available at each turn is any non-negative number and the amount allocated to vertices lies between 0 and 1. A vertex with a protection less than~1 is {\em partially protected} and its unprotected part can burn partially and transmit only its fraction of fire  to the adjacent vertices. Thus, the $f_i$ may take any non-negative value. This defines a variant game called {\sc Fractional Firefighter} which was introduced in~\cite{fogarty}.

\subsection{Our contribution}

The main thread of this paper is the focus on general firefighter sequences, which raises three specific research questions. We address these questions when a single fire spreads throughout a rooted tree.

First, we introduce an online version of both {\sc Firefighter} and {\sc Fractional firefighter}  where the sequence of firefighters is revealed over time (online) while the graph (a tree in our case) is known from the start.  
To our knowledge, this is the first attempt at analysing online firefighter problems. Although our motivation is mainly theoretical, this paradigm is particularly natural in emergency management where one has to make quick decisions despite lack of information. Any progress in this direction tells us how lack of information impacts the quality of the solution. 
Note that a version of the game introduced in~\cite{bonato12} also models a lack of information. In that version, rather than the firefighting resources, the missing information is where the fire will spread. Also, they propose randomised analyses to maximise the expected number of saved vertices while we use worst case analyses expressed in terms of competitive ratios.

A second question, the {\em separating problem}, deals specifically with infinite trees.  
Separating two given firefighter sequences means finding an infinite tree on which the fire can be contained with one sequence but not the other.

The third question deals with criteria for the fire to be contained based on the asymptotic behaviours of the firefighter sequence and the size of the levels in the tree. Unlike the first two questions, it has already been investigated in other papers (e.g., \cite{Dyer-containment,Lehner}) for {\sc Firefighter} with firefighter sequences of the form $(\lambda^n)$.

The paper is organised as follows: 
in~\Cref{sec:problems} we define formally {\sc Firefighter} and {\sc Fractional Firefighter} as well as their online versions.  
\Cref{sec:newtrees} deals with competitive analysis when the fire spreads in a finite tree and the firefighter sequence is revealed online. We first generalise an analysis of a greedy algorithm  known only in special cases of {\sc Firefighter} to {\sc Fractional Firefighter}. For the offline case, it answers an open question proposed in~\cite{Finbow.MacGillivray,Hartke.Stephen}. Then we propose improved competitive algorithms for online {\sc Firefighter} with a small total number of firefighters while establishing that the greedy approach is optimal in the general case.  
The last two sections (\Cref{sec:separe} and \Cref{sec:asymptotic}) both deal with the infinite case. \Cref{sec:separe} deals with our second question. Considering the class of spherically symmetric trees where all vertices at the same level have the same degree, we express the separation problem as a purely numerical one-player game, which we call the {\em targeting game}. We propose two sufficient conditions for the existence  of a winning strategy.  
\Cref{sec:asymptotic} deals with our third question. We establish sufficient conditions for containing the fire  expressed as asymptotic comparisons of the number of available firefighters and the size of the levels in the tree. In the online case, for a particular class of trees the level size of which grows linearly, we also give a sufficient condition to contain the fire.

\subsection{Some notations}\label{sec:notations}

Given a tree $T$ rooted in $r$, $V(T)$ and $E(T)$ will denote the vertex set and the edge set of $T$, respectively. Given two vertices $v$ and $v'$, $v\lhd  v'$ denotes that $v$ is an ancestor of $v'$ (or $v'$ is a descendant of $v$) and $v\unlhd v'$ denotes that either $v=v'$ or $v \lhd v'$. For any vertex $v$, let $T[v]$ denote the sub-tree induced by $v$ and its descendants. Let $T_i$ denote the $i$-th level of $T$ rooted in $r$, where $\{r\}=T_0$.
For a finite tree $T$ rooted in $r$, the height  $h(T)$ is the maximum length of a path from $r$ to a leaf. If $i>h(T)$, we have $T_i=\emptyset$. The \emph{weight} $w(v)$ of a vertex $v$ is the number of vertices of $T[v]$. 
When no ambiguity may occur, we will simply write $w_v=w(v)$.
 
We denote by $B(T)$ the tree obtained from $T$ by contracting all vertices from levels $0$ and $1$  into a new root vertex $r_B$: for all $u_1 \in T_1$ and $u_2 \in T_2$, every edge $ru_1$ is contracted and every edge $u_1u_2 \in E(T)$ gives rise to an edge $r_Bu_2 \in E(B(T))$. For $k\leq h(T)$,   $B^k(T)$ will denote the $k^{\rm th}$ iteration of $B$ applied to $T$: all vertices from levels $0$ to $k$ are contracted into a single vertex denoted by $r_{B^k}$ which becomes the new root. 

Given a predicate $P$, we denote by $\mathds 1_{P}$ the associated characteristic function so that $\mathds 1_{P(x)}=1$ if $P(x)$ is true and $0$ otherwise.

\section{Problems and preliminary results}\label{sec:problems}

\subsection{{\sc Firefighter} and {\sc Fractional Firefighter}}

An instance of the {\sc Fractional Firefighter}
is defined by a triple $(G,r,(f_i))$, where $G=(V(G),E(G))$ is a finite graph, $r\in V(G)$ is the vertex where the fire breaks out and $(f_i)_{i\geq 1}$ is the non-negative {\em firefighter sequence}, ($f_i$ indicates the amount of protection that can be placed at turn $i$). Note that the game could be extended by allowing negative values for $f_i$, however, we will exclude pyromaniac firefighters from this paper, with one exception in \Cref{SST} for the purpose of simplifying a proof. Let $S_i$ denote the cumulative amount of firefighters received: $S_i=\sum_{j=1}^i f_j$.

Turn $i=0$ is the initial state where $r$ is burning and all other vertices are unprotected, and $i\geq 1$ corresponds to the different rounds of the game. At each turn $i\geq 1$ and for every vertex $v$, the player decides which amount $p_i(v)$ of protection to add to $v$, with $0\leq p_i(v)\leq 1$. Throughout the game, for every vertex $v$ the part of  $v$ which is burning at turn $i\geq 0$ is denoted by $b_i(v)$, with $0\leq b_i(v)\leq 1$, $b_0(r)=1$ and $b_0(v)=0$ for all $v\neq r$. Similarly the cumulative protection received by vertex $v$ is $p^c_i(v)$ with $p^c_0(v)=0$ for all $v$. Both $(b_i(v))_{i\geq 0}$ and  $(p^c_i(v))_{i\geq 0}$ are non-decreasing sequences with $b_i(v)+p^c_i(v)\leq 1$ for all $i$ and $v$. At each turn $i$, the player's choice of $p_i(v)$ is subject to the constraints $p_i(v)\geq 0$, $b_{i-1}(v)+p^c_{i-1}(v)+p_i(v)\leq 1$ and $\sum_{v\in V(G)}p_i(v)\leq f_i$. The new protection of $v$ is $p^c_i(v)=p^c_{i-1}(v)+p_i(v)$. The fire then spreads following the rule 
\begin{equation}\label{eq:fire spread}
b_{i}(v)=\max \{\max_{v'\in N(v)} b_{i-1}(v')-p^c_i(v),b_{i-1}(v)\},
\end{equation}
where $N(v)$ denotes the open neighbourhood of $v$. The game finishes when the fire stops spreading (i.e. $b_{i}(v)=b_{i-1}(v)$ for all $v$). 

The standard {\sc Firefighter} problem is similar to {\sc Fractional Firefighter}, but with the additional constraint that the $p_i(v)$ are all binary variables. It follows that $p^c_i(v)$ and $b_i(v)$ are also binary. In this case, we require that the firefighter sequence has integral values.

We will now show that both versions of the game always terminate on a finite graph $G$. Let $L_r(G)$ denote the maximum length of an induced path in $G$ with extremity $r$, we have:  

\begin{proposition}\label{prop:terminates}
The maximum number of turns before a game of {\sc Firefighter} or {\sc Fractional Firefighter} on a finite graph $G$ will terminate is $L_r(G)$.
\end{proposition}

\begin{proof}

First,  
we show by induction on $i$ that: 

\begin{quote}
For all $i\geq 1$, for all vertex $v$, if $b_{i}(v)>b_{i-1}(v)$, then there is an induced path $P_{i,v}=(u_0=r,u_1,\ldots,u_i=v)$ of length $i$ such that $b_{j}(u_j)$ is non-increasing along the path.
\end{quote}

For $i=1$ and $v\in V(G)$, if $b_{1}(v)>b_0(v)$  then $v$ is a neighbour of $r$ and $b_0(r)=1\geq b_1(v)$. The path $P_{1,v}=(r,v)$ is of length 1.

Suppose the property holds at turn $i$ and that $b_{i+1}(v)>b_i(v)$ for some $v$. Necessarily, $v$ receives the additional amount of fire from a neighbour $w$: 
$$b_{i+1}(v)=b_i(w)-p^c_{i+1}(v)>b_i(v).$$ 

We necessarily have $b_i(w)>b_{i-1}(w)$ since otherwise $b_i(w)=b_{i-1}(w)$ and we would have the following contradiction: 
$$b_i(v)\geq b_{i-1}(w)-p^c_i(v)= b_{i}(w)-p^c_i(v)
\geq b_{i}(w)-p^c_{i+1}(v)>b_i(v).$$

Applying to $w$ the induction hypothesis, there is an induced path $P_{i,w}=(u_0=r,u_1,\ldots,u_i=w)$ such that $(b_j(u_j))$ is non-increasing. Since, $b_{i+1}(v)=b_i(w)-p^c_{i+1}(v)$, we have $b_{i+1}(v)\leq b_i(w)$.

Also, for all $j<i$, $u_j$ and $v$ are not adjacent since otherwise, we would have the following contradiction:

$$b_i(v)\geq b_{j+1}(v)\geq b_j(u_j)-p^c_{j+1}(v)\geq b_i(w)-p^c_{i+1}(v)>b_i(v).$$

It follows that the path $P_{i+1,v}$ obtained by adding the edge $wv$ to $P_{i,w}$ is an induced path which satisfies the required property.\\ 
Thus, if $b_{i+1}(v)>b_i(v)$ for some $v$, we have $i+1\leq L_r(G)$. Consequently, the fire can no longer spread at turn $L_r(G)+1$.

\bigskip
Conversely, let $P$ be an induced path with extremity $r$ of length $L_r(G)$. If $(f_i)=(|G|-L_r(G)-1,0,0,0 \cdots)$ and if the complement of $P$ is protected during the first turn, the game will terminate in exactly $L_r(G)$ turns.
\end{proof}

\begin{example}
On a perfect binary tree $(B_n,r)$ of height $n$,  given one firefighter per turn, the length of {\sc (Fractional) Firefighter} is exactly $n=L_r(B_n)$, whatever the player's strategy (this will be an immediate consequence of \Cref{spherically symmetric trees}).
\end{example}

\subsection{Simplification for Trees}\label{sec:simple-tree}

In this paper, we focus on the case of trees. Given an instance $(T,r,(f_i))$, where $T$ is a tree, $T$ will be considered rooted in $r$. In order to remove trivial cases, we will exclude algorithms which place at turn $i$ more protection on a vertex $v$ than the part of $v$ that would burn if no protection were placed starting from turn $i$.  If $T$ is finite, \Cref{prop:terminates} implies that the game will end  in at most $h(T)$ turns. We consider that it has exactly $h(T)$ turns, eventually with empty turns where no firefighters are allocated towards the end of the game.

Solutions on trees have a very specific structure.
Indeed, it immediately follows from \Cref{eq:fire spread} that, when playing on a tree,
at each turn $i$, the amounts of fire $b_i(v)$ are non-increasing along any path from the root, which means that the fire will only spread outwards from the root.  
Also, for every vertex $v$ in $T_j$ for some $j$, the amount of fire $b_j(v)$ can no longer increase after turn $j$. Hence,  
no protection is placed in $T_j$ at turn $i>j$. Note also that for any  
solution  
which allocates a positive amount of protection at turn $i$ to a vertex $v\in T_k, k>i$, allocating the same amount of protection to the parent of $v$ instead strictly improves the performance. Indeed, if vertex $v$ can still burn, so can its parent.  So we may consider only  algorithms that play in $T_i$ at turn $i$.  For an optimal  
algorithm, this property was emphasised in~\cite{Hartnell.Li}. 

This holds for both {\sc Firefighter} and {\sc Fractional Firefighter} on trees. For such an algorithm, $p^c_i(v)=p_i(v)$ and the values of $p_i(v)$ and $b_i(v)$ will not change after turn $i$ for $v\in T_i$. Hence, the index $i$ may be dropped by denoting  $p(v)=p_i(v)$ and $b(v)=b_i(v)$ for $v\in T_i$. A solution $p$ is then characterised by the values $p(v), v\in V(T)$. 
For any solution $p$, while $p(v)$ represents the amount of protection received directly, vertex $v$ also receives protection through its ancestors, the amount of which is denoted by $P_p(v)=\sum_{v'\lhd v} p(v')$ (used in \Cref{sec:greedy}). Since we only consider algorithms that play in $T_i$ at turn $i$ and do not place extraneous protection, for any vertex $v$, $p(v)+P_p(v)\leq 1$. Also, for any vertex $v\in T_i$, we have $b_{i-1}(v)=0$ and by summing  \Cref{eq:fire spread} from $1$ to $i$, we deduce that $p(v)+P_p(v)+b(v)=1$. 

Any  
solution $p$ for {\sc Firefighter} or {\sc Fractional Firefighter} will satisfy the constraints:
\begin{center}
 $[{\cal C}]\left\{\begin{array}{ll}
\sum_{v\in T_i} p(v)\leq f_i& \;\; (i) \\
\forall v, p(v)+P_p(v)\leq 1 & \;\; (ii)
\end{array}\right.$
\end{center}

In~\cite{MacGillivray.wang}, a specific boolean linear model has been proposed for solving {\sc Firefighter} on a tree $T$ involving these constraints. Solving {\sc Fractional Firefighter} on $T$ corresponds to solving the relaxed version of this linear program.

\subsection{Online version}\label{sec:online-version}

Online optimisation~\cite{paschos-online} is a generalisation of approximation theory which represents situations where the information is revealed over time and one needs to make irrevocable decisions. Following the definitions of~\cite{albers},
we now introduce  online versions of {\sc Firefighter} and 
{\sc Fractional Firefighter} on trees. The graph and starting point of the fire are known from the start by the player,  
 but the firefighter sequence $(f_i)_{i\geq 1}$ is revealed over time, i.e., at step $i$ the player does not know any request $f_j$ with $j>i$. This set-up can be seen as a game between the online player (or algorithm) and an  {\em oblivious adversary}.     
At each turn $i$, the oblivious adversary reveals $f_i$ and then the player chooses where to allocate this resource.  We refer to the usual case, where $(f_i)_{i\geq 1}$ is known in advance by the player, as {\em offline}.

Let us consider an online algorithm $\mathbf{OA}$ for one of the two problems and let us play the game on a finite tree $T$ until the fire stops spreading. The value $\lambda_{\mathbf{OA}}$ achieved by the algorithm, defined as the amount of saved vertices, is measured against the best value performed by an algorithm which knows in advance the sequence $(f_i)$. In the present case, it is simply the optimal value of the offline instance, referred to as the {\em offline optimal value}, denoted by  $\beta_I$ when considering the online {\sc  Firefighter} ($I$ stands for ``Integral") and  $\beta_F$ for the online {\sc  Fractional Firefighter}. We will call $\mathbf{Bob}$ such an algorithm, able to see the future and guaranteeing the value $\beta_I$ or $\beta_F$ for online {\sc  Firefighter} and {\sc Fractional Firefighter}.   

Algorithm $\mathbf{OA}$ is said to be $\gamma$-competitive, $\gamma\in ]0,1]$,  for the online {\sc  Firefighter} (resp. {\sc  Fractional Firefighter}) if for every instance, $\frac{\lambda_{\mathbf{OA}}}{\beta_I}\geq \gamma$ (resp. $\frac{\lambda_{\mathbf{OA}}}{\beta_F}\geq \gamma$);  
$\gamma$ is also called a {\em competitive ratio} guaranteed by $\mathbf{OA}$. An online algorithm will be called {\em optimal} if it guarantees the best possible competitive ratio. It should be noted that to calculate the competitive ratio, it is necessary to evaluate the worst case of all the adversary's possible choices. This is actually like considering a malicious adversary instead of an ignorant one.

The arguments given in \Cref{sec:simple-tree}, to justify considering only algorithms that play in $T_i$ at turn $i$, are still valid for online algorithms; thus we will only consider such online algorithms. Let us start with a reduction:

\begin{proposition}\label{first turn} We can reduce online {\sc (Fractional) Firefighter} on trees to  instances where $f_1>0$.
\end{proposition}

\begin{proof}  If $f_i=0$ for all $i$ such that $1\leq i\leq k$, then the instance $(T,r,(f_i))$ is equivalent to the instance $(B^k(T),r_{B^k},(f_{i+k}))$.
\end{proof}

In the infinite case we do not define competitive ratios, but only ask whether the fire can be contained by an online algorithm. \Cref{sec:greedy,sec:improved} deal with the finite case
while \Cref{sec:linear} deals with a class of infinite trees.

\section{Online Firefighting on Finite Trees}\label{sec:newtrees}

\subsection{Competitive analysis of a Greedy algorithm}\label{sec:greedy}

Greedy algorithms are usually very good candidates for online algorithms, sometimes the only known approach.
Mainly two different greedy algorithms have been considered in the literature for {\sc Firefighter} on a tree~\cite{Finbow.MacGillivray} and they are both possible online strategies in our set-up. The {\em degree greedy} strategy prioritises saving  
vertices of large degree; it has been shown in~\cite{bazgan-chopin-ries}  that  
it cannot guarantee any approximation ratio 
on trees, even for a constant firefighter sequence. 
A second greedy algorithm, introduced in~\cite{Hartnell.Li} for an integral sequence $(f_i)$,   
maximises at each turn the total weight of the newly protected vertices. We generalise it to any firefighter sequence  
for both the integral and the fractional problems. 
Let $\mathbf{GR}$ denote the greedy algorithm that selects at each turn $i$ an optimal solution of the linear program ${\cal P}_i$ with variables $x(v), v\in T_i$ and constraints $[{\cal C}]$:
 
$${\cal P}_i:\left\{\begin{array}{ll}
\max \sum_{v\in T_i} x(v)w(v)& \\
\sum_{v\in T_i} x(v)\leq f_i& \;\; (i) \\
\forall v\in T_i, x(v)+P_x(v)\leq 1 & \;\; (ii)
\end{array}\right.$$

An optimal solution of ${\cal P}_i$ is obtained by ordering vertices $\{v_1, \ldots, v_{|T_i|}\}$ of level $i$ by non-increasing weight and taking them one by one in this order and greedily assigning to vertex $v_j$ the value $x(v_j)=\min(f_i-\sum_{k<j}x(v_k),1-P_x(v_j))$. 
Note that $\mathbf{GR}$ is valid for both {\sc Firefighter} and {\sc Fractional Firefighter}.

It was shown in \cite{Hartnell.Li} that the greedy algorithm on trees gives a $\frac 12$-approximation of the restriction of {\sc Firefighter} when a single firefighter is available at each turn. They claim that this approximation ratio remains valid for a fixed number $D\in \mathbb{N}$ of firefighters at each turn. We extend this result to any firefighter sequence $(f_i)_{i\geq 1}$, integral or not. Since $\mathbf{GR}$ is an online algorithm, the performance can also be seen as a competitive ratio for the online version.

\begin{theorem}\label{prop:greed}
The greedy algorithm $\mathbf{GR}$ is $\frac 12$-competitive for  both online {\sc Firefighter} and {\sc Fractional Firefighter} on finite trees.
\end{theorem}

\begin{proof}

Let us first consider the fractional case with an online instance $(T,r,(f_i))$ of {\sc Fractional Firefighter} on a tree. 

Let $x(v)$ and $y(v)$ be the amounts of firefighters placed on vertex $v$ by $\mathbf{GR}$ and $\mathbf{Bob}$, respectively. We have $\lambda_{\mathbf{GR}}=\sum_v x(v)w(v)$ and $\beta_F=\sum_v y(v)w(v)$. 

Recall that $P_x(v)=\sum_{v'\lhd v} x(v')$ and  $P_y(v)=\sum_{v'\lhd v} y(v')$.
We split $y(v)$ into two non-negative quantities, $y(v)=g(v)+h(v)$, where $g(v)$ is the part of $y(v)$ already protected by $\mathbf {GR}$ through the ancestors of $v$, while $h(v)$ is the part of $y(v)$ which, when added on top of $P_y(v)$, exceeds $P_x(v)$. So, if $P_y(v)<P_x(v)<y+P_y(v)$, we have $h(v)=y(v)-P_x(v)$. The general formula is: \\ $g(v)=\min\{y(v),\max\{0,P_x(v)-P_y(v)\}\}$ and $h(v)=\max\{0,y(v)+\min\{0,P_y(v)-P_x(v)\}\}.$

We now claim that $\forall v'\in T$, $\sum_{v\unlhd v'} g(v) \leq P_x(v')$ and prove it by induction.
Since $g(r)=0$, it holds for the root $r$. Assuming that the inequality holds for a vertex $v'$, let $v''$ be a child of $v'$. If $P_x(v'')-P_y(v'')\geq 0$, then we directly have:
$$\sum_{v\unlhd v''} g(v)=\sum_{v\lhd  v''} g(v) + g(v'') \leq \sum_{v\lhd  v''} y(v) + (P_x(v'')-P_y(v''))=P_x(v'').$$

Else, $g(v'')=0$ and using  $\sum_{v\unlhd v'} g(v) \leq P_x(v')$ 
  and  $P_x(v'')\geq P_x(v')$, the inequality holds for $v''$; which completes the proof of the claim. Thus:
$$ \sum_{v'}\sum_{v\unlhd v'} g(v)\leq \sum_{v'} P_x(v')= \sum_{v'}\sum_{v\lhd v'} x(v)\leq \sum_{v'}\sum_{v\unlhd v'} x(v).$$

Since $w(v)=\sum_{v\unlhd v'} 1$, by changing the order of summation on both sides, we obtain:

\begin{equation}\label{eq1} \sum_{v} g(v)w(v)\leq \sum_v x(v)w(v)=\lambda_{\mathbf{GR}}.
\end{equation}

Let us now consider the coefficients $h(v)$. We claim that the coefficients $h(v)$ with $v\in T_i$  satisfy  the constraints $(i)$ and $(ii)$ of ${\cal P}_i$: indeed for $(i)$, we have $h(v)\leq y(v)$ and $y$ satisfies constraint $(i)$. For $(ii)$ note that $h(v)+P_x(v)=\max\{P_x(v),y(v)+\min\{P_x(v),P_y(v)\}\}\leq \max\{P_x(v),y(v)+P_y(v)\}\leq 1$.

Hence, $\forall i, \sum_{v\in T_i} h(v)w(v)\leq \sum_{v\in T_i} x(v)w(v)$ and therefore:
\begin{equation}\label{eq2} \sum_{v\in T} h(v)w(v)\leq \sum_{v\in T} x(v)w(v)=\lambda_{\mathbf{GR}}.\end{equation}

Finally, since $g(v)+h(v)=y(v)$, we conclude from \Cref{eq1,eq2} that $\beta_F\leq 2\lambda_{\mathbf{GR}}$.\\
Hence the Greedy algorithm is $\frac 12$-competitive for the online {\sc Fractional Firefighter} problem. Since the greedy algorithm gives an integral solution if $(f_i)$ has integral values and since $\beta_F\geq \beta_I$, it is also $\frac 12$-competitive for the {\sc Firefighter} problem. This concludes the proof of \Cref{prop:greed}.

\end{proof}

Conjecture 2.3 in~\cite{Hartke.Stephen}   (which is also Conjecture 3.5 in~\cite{Finbow.MacGillivray}) claims that there is a constant $\rho$ such that the optimal value of {\sc Fractional Firefighter} on trees is at most $\rho$ times the optimal value of {\sc Firefighter}. It was supported by extensive experimental tests \cite{Hartke.Stephen}, but finding such a constant and proving the ratio is one of the open problems proposed in~\cite{Finbow.MacGillivray} (Problem 7). \Cref{prop:greed} can be expressed as 
$\lambda_{\mathbf{GR}} \leq \beta_I \leq \beta_F \leq 2\lambda_{\mathbf{GR}}$,
which shows that $\rho=2$ is such a constant:

\begin{corollary}\label{OI/OF} In {\sc Fractional Firefighter} on trees, the amount of vertices saved is at most twice  the maximum number of vertices saved in {\sc Firefighter}.
\end{corollary}

\subsection{Improved Competitive Algorithm for {\sc Firefighter}}\label{sec:improved}

In this section,  
we investigate possible improvements for online strategies for {\sc Firefighter} on finite trees. Let $\varphi=\frac{1+\sqrt 5}{2}$ denote the golden ratio, satisfying 
$\varphi^2=\varphi+1$.	

For any integer $k\geq 2$, we denote by 
$\alpha_{I,k}$ the best possible competitive ratio for online {\sc Firefighter} on finite trees if at most $k$ firefighters are available in the entire game. We have:
$$\alpha_{I,k}=\mathop{\rm inf}\limits_{T\in {\cal T}}\  \mathop{\rm max}\limits_{\mathbf{OA}\in {\cal A}_L}\  \mathop{\rm min}\limits_{(f_i)\in \mathbb N^{\mathbb N^*}, \sum_if_i\leq k } \ \frac{\lambda_{\mathbf{OA}}}{\beta_I},$$ 
where ${\cal T}$ denotes the set of finite rooted trees, ${\cal A}_L$ the set of online algorithms for {\sc Firefighter} on finite trees and $\mathbb N^*$ denotes the set of positive integers.

Note that in the definition of $\alpha_{I,k}$, $\lambda_{\mathbf{OA}}$ and $\beta_I$ depend on $T$. Also, the maximum and the minimum are  well defined since on a finite tree $T$, the set of possible ratios is finite. An online algorithm, choosing for any fixed $T$ a strategy which achieves this maximum, will be $\alpha_{I,k}$-competitive for instances  with at most $k$ firefighters. Such an algorithm is optimal for these instances.

The sequence $\left(\alpha_{I,k} \right)$ is non-increasing. 
We define $\alpha_I=\lim\limits_{k\rightarrow\infty}\alpha_{I,k}$; again, the index $I$ stands for {\em Integral} and refers to the {\sc Firefighter} problem.

\begin{remark}
The limit $\alpha_I$ is the greatest competitive ratio that can be reached on any tree. 
Indeed, given a finite tree $T$, it suffices to consider the instances with at most $|V(T)|$ firefighters.
\end{remark}

In this section, we give an online algorithm for instances of {\sc Firefighter} on a finite tree that is optimal (i.e., $\alpha_{I,2}$-competitive)  
if at most two firefighters are presented.
Based on \Cref{first turn}, we may assume $f_1\neq 0$.   
If $f_1 = 2$, one firefighter will be called the first and the other one the second.
An online instance is then characterised by 
when the second firefighter is presented. It can be never if only one firefighter is presented or at the first turn if $f_1=2$. Note that this later case is trivial since an online algorithm can make the same decision as $\mathbf{Bob}$ by assigning both firefighters to  two unburnt vertices of maximum weights. Our algorithm works also in this case and will make this optimal decision. 

\begin{lemma}\label{claim:ab}
Let $a$ and $b$ be two vertices of maximum weights in $T_1$.
If $\sum_if_i\leq 2$, there is an optimal offline algorithm for {\sc Firefighter} which places the first firefighter on either $a$ or $b$.  
\end{lemma}

\begin{proof}
If the first firefighter is placed on $v\in T_1\setminus \{a,b\}$ by an optimal offline algorithm, since at most two firefighters are available,  
$\exists u\in \{a,b\}$, $T[u]$ burns completely.  
Hence, replacing $v$ by $u$ when assigning the first firefighter would produce another optimal solution (necessarily $w_v=w_u$).   
\end{proof}

We suppose $\mathbf{Bob}$ has this property. However, even if $w_a>w_b$, he will not necessarily choose $a$; as illustrated by the graph $W_{1,10,20}$ (\Cref{pincer}) where if the firefighter sequence is $(1,0,1,0,0,0\ldots)$, then $\mathbf{Bob}$'s needs to protect $x$ during the first turn. 
Note also that, when the root is of degree at least~3,  the second firefighter is not necessarily in $V(T[a])\cup V(T[b])$.

We now consider \Cref{algo:2firefighters} and assume that the adversary will reveal at most two firefighters.  
The algorithm works on an updated version $\widetilde T$ of the tree: if one vertex is protected, then the corresponding sub-tree is removed  
and all the burnt vertices are contracted into the new root $\tilde r$ so that the algorithm always considers vertices of level 1 in 
$\widetilde T$.  
Before starting the online process, the algorithm computes the weights of all vertices.  
The weights of the unburnt vertices do not change when updating $\widetilde T$.   
The value of $h(\widetilde T)$, required in line~\ref{line:Greedy test}, can be computed during the initial calculation of weights and easily updated  with $\widetilde T$. For the sake of clarity, we do not detail all updates in the algorithm. 

In this section only, for any vertex $v\in T_i$ and any $i\leq j \leq i+h(T[v])$, we denote by $v_j$ a vertex of maximum weight $w_{v_j}$ in $T_j\cap V(T[v])$, i.e. among the descendants of $v$ which are in level $j$ (or $v$ itself if $i=j$). We also define $\bar w_{v_j}$ for all $j$ via:
$$
\bar w_{v_j}=w_{v_j}\ {\rm if}\ j\in [i; i+h(T[v])]\ {\rm and}\ 0\ {\rm otherwise.}
$$
 
\begin{algorithm}[ht] 
\floatname{algorithm}{Algorithm}
\caption{}
\begin{algorithmic}[1]
\label{algo:2firefighters}
\REQUIRE{A finite tree $T$ with root $r$ \-- An online adversary.}
\STATE{$(\widetilde T,\tilde r)\leftarrow (T,r)$; Compute $w_v$, $\forall v\in V(\widetilde T)$}
\STATE{$First\_Firefighter \leftarrow TRUE$;}
\STATE{}\COMMENT{Start of the online process}
\STATE{At each turn, after the fire spreads, $\widetilde T$ is updated \-- burnt vertices are contracted to $\tilde{r}$;}
\STATE{If several firefighters are presented at the same time, we consider them one by one in the following lines;}
\IF{a new Firefighter is presented \AND $\tilde r$ has at least one child}
\IF{$First\_Firefighter$}
\STATE{Let $a$ and $b$ denote two  children of $\tilde r$ with  maximum weight $w_a,w_b$ and $w_a\geq w_b$ ($a=b$ if $\tilde r$ has only one child);}\label{line:ab}
\IF{$\min\limits_{2\leq i\leq 1+h(\widetilde T)}\frac{w_a+\bar w_{b_i}}{w_b+\bar w_{a_i}}\geq\frac 1\varphi$}\label{line:Greedy test}
		\STATE{Place the first firefighter on $a$;}\label{line:on a} 
		\ELSE
		\STATE{Place the first firefighter on $b$;} \label{line:on b} 
		\ENDIF
        \STATE{$First\_Firefighter \leftarrow FALSE$;}
\ELSE
\STATE{Place the firefighter on a child $v$ of $\tilde r$} of maximum weight\label{line:second firefighter}
\ENDIF
\ENDIF
\end{algorithmic}
\end{algorithm}

	\begin{theorem} \label{theo2}\Cref{algo:2firefighters} is a $\frac 1\varphi$-competitive online algorithm for online {\sc Firefighter} with at most two firefighters available. It is optimal for this case.
\end{theorem}

\begin{proof}

While \Cref{algo:2firefighters} runs feasibly on any instance, we limit the analysis to the case where at most two firefighters are available.
 If the adversary does not present any firefighter before the turn $h(T)$, both \Cref{algo:2firefighters} and  
 $\mathbf{Bob}$ cannot save any vertex and, by convention, the competitive ratio is~1. 
 
 Let us suppose that at least one firefighter is presented at some turn $k\leq h(T)$;  
 the tree still has at least one unburnt vertex. During the first $(k-1)$ turns,  the instance is updated to $(B^{k-1}(T),r_{B^{k-1}},(f_{i+k-1}))$. In the updated instance, at least one firefighter is presented during the first turn and the root has at least one child. \Cref{first turn} ensures that it is equivalent to the original instance.

If the root $r_{B^{k-1}}$ has only one child, line~\ref{line:ab} gives $a=b$ and \Cref{algo:2firefighters} selects $a$ at line~\ref{line:on a}. In the updated instance, all vertices are saved; so the competitive ratio is equal to 1.

Else, we have $a\neq b$ with $w_a\geq w_b$ (line~\ref{line:ab}). If the adversary presents a single firefighter for the whole game, then \Cref{algo:2firefighters} protects either $a$ or $b$. Meanwhile,  
 $\mathbf{Bob}$ will protect $a$, saving $w_a$ vertices. If $w_a\geq \varphi w_b$, then we have: 
	\begin{equation}\label{eq:big a}
	\forall i,  2\leq i\leq 1+h(T),  \frac{w_a+\bar w_{b_i}}{w_b+\bar w_{a_i}}\geq \frac{w_a}{w_b+w_a}\geq  \frac{\varphi w_b}{w_b+\varphi w_b}=\frac 1\varphi.
\end{equation}
So \Cref{algo:2firefighters} protects $a$ (line~\ref{line:Greedy test}), guaranteeing a competitive ratio of 1. Otherwise, if $w_b>\frac 1\varphi w_a$, even placing the firefighter on $b$ guarantees a ratio of at least $\frac 1\varphi$.

Suppose now that the adversary presents two firefighters.   
We consider two cases.\\
{\underline {Case (i)}}: If \Cref{algo:2firefighters} places the first firefighter on $a$ at line~\ref{line:on a}, and if the adversary presents the second firefighter at turn $i\geq k$, then the algorithm will save $w_a+\bar w_{x_i}$, for some $x\in T_{k}\setminus\{a\}$ such that 
 $\bar w_{x_i}=\max_{u\in T_k\setminus\{a\}}\bar w_{u_i}.$  
 For the same instance, $\mathbf{Bob}$  will save $w_v+\bar w_{y_i}$ for some $v\in \{a,b\}$ and $y\in T_{k}\setminus\{v\}$. If the two values are different (the optimal one is strictly better), then necessarily  $v=b$ and $y=a$. In this case the criterion of line~\ref{line:Greedy test} ensures that the related competitive ratio is at least $\frac 1\varphi$.\\
{\underline {Case (ii)}}: Suppose now that \Cref{algo:2firefighters} places the first firefighter on $b$ at line~\ref{line:on b}, and say the adversary presents the second firefighter at turn $j\geq k$. Lines~\ref{line:ab} and \ref{line:Greedy test} ensure that:  
\begin{equation}\label{eq:b chosen}
	\exists i,2\leq i\leq 1+h(T),  \frac{w_a+\bar w_{b_i}}{w_b+\bar w_{a_i}}<\frac 1\varphi.
\end{equation}
 Hence, we have $w_a<\varphi w_b$, since in the opposite case, \Cref{eq:big a} would hold.  
\Cref{algo:2firefighters} now saves $w_b+\bar w_{x_j}$, for $x\in T_{k}\setminus\{b\}$ such that 
$\bar w_{x_j}=\max_{u\in T_{k}\setminus\{b\}}\bar w_{u_j}$.
Meanwhile, 
 $\mathbf{Bob}$  selects $v\in \{a,b\}$ and, if it exists, $y_j$ for some $y\in T_{k}\setminus\{v\}$, for a total of $w_v+\bar w_{y_j}$ vertices saved.
If $y\neq b$, then $\bar w_{y_j}\leq \bar w_{x_j}$, by definition of $x$, and thus: 
 
 \begin{equation}\label{eq:xu not from b}
 	\frac{w_b+\bar w_{x_j}}{w_a+\bar w_{y_j}}\geq \frac{w_b+\bar w_{x_j}}{w_a+\bar w_{x_j}}\geq \frac{w_b}{w_a}>\frac 1\varphi.
 \end{equation}

Finally, if $y=b$, then $v=a$ and the competitive ratio to evaluate is $ \frac{w_b+\bar w_{x_j}}{w_a+\bar w_{b_j}} $. We claim that the following  holds:
  
  \begin{equation}\label{eq:double ratio}
  	\frac{w_a+\bar w_{b_i}}{w_b+\bar w_{a_i}}\times \frac{w_b+\bar w_{x_j}}{w_a+\bar w_{b_j}}\geq \frac{1}{\varphi^2}.
  \end{equation}

If $i\geq j$, then $\bar w_{a_i}\leq \bar w_{a_j}$ and since $a\neq b$, 
$\bar w_{a_j}\leq \bar w_{x_j}$. Hence:
$
\frac{w_b+\bar w_{x_j}}{w_b+\bar w_{a_i}}\geq \frac{w_b+\bar w_{x_j}}{w_b+\bar w_{a_j}}\geq 1
$ 
  and therefore:
 $$
\frac{w_a+\bar w_{b_i}}{w_b+\bar w_{a_i}}\times \frac{w_b+\bar w_{x_j}}{w_a+\bar w_{b_j}}\geq \frac{w_a+\bar w_{b_i}}{w_a+\bar w_{b_j}}\geq \frac{w_a}{w_a+w_b}.  
  $$
  
Now, if $i<j$, we get:
 $
 \frac{w_a+\bar w_{b_i}}{w_a+\bar w_{b_j}}\geq 1
 $
  and therefore:
 $$
\frac{w_a+\bar w_{b_i}}{w_b+\bar w_{a_i}}\times \frac{w_b+\bar w_{x_j}}{w_a+\bar w_{b_j}}\geq \frac{w_b+\bar w_{x_j}}{w_b+\bar w_{a_i}}  \geq \frac{w_b}{w_a+w_b}.  
  $$ 
  
In both cases, since $\frac{w_a}{w_a+w_b}\geq \frac{w_b}{w_a+w_b}\geq \frac{1}{1+\varphi}=\frac{1}{\varphi^2}$, we obtain  \Cref{eq:double ratio}. 
Now, \Cref{eq:double ratio,eq:b chosen} imply that in case (ii), when $y=b$, we also have 
$
\frac{w_b+\bar w_{x_j}}{w_a+\bar w_{b_j}}\geq \frac{1}{\varphi}.$
Together with \Cref{eq:xu not from b}, this concludes  case (ii) and shows that \Cref{algo:2firefighters} is $\frac 1\varphi$-competitive.
\end{proof}

Even though complexity analyses are not usually proposed for online algorithms, it is worth  noting that line~\ref{line:Greedy test} only requires the weights of vertices in $V(T[a])\cup V(T[b])$ and the maximum weight per level in $T[a]$ and $T[b]$. Hence, \Cref{algo:2firefighters} requires $O(|V(T[a])|+|V(T[b])|)$ to choose the position of the first firefighter and $O(|V(T)|)$ altogether.

We conclude this section with a hardness result justifying that the greedy algorithm $\mathbf{GR}$ is optimal and that \Cref{algo:2firefighters} is optimal if at most two firefighters are available. These hardness results will all be derived from the graphs $W_{k,l,m}$ (\Cref{pincer}).

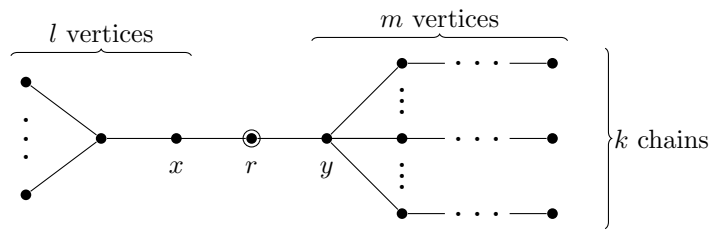
\begin{figure}[htb]
\begin{center}
	\begin{tikzpicture}

		\node[fill,circle,inner sep = 0.05 cm] (r) at (0,0) {};
		\node[draw,circle,inner sep=0.08cm] at (0,0) {};
		\node[below = 0.1 cm of r] () {$r$};

		\node[fill,circle,inner sep = 0.05 cm] (arbre) at (-2,0) {};
		\draw (arbre) to (r);
		\node[below = 0.1 cm of arbre] () {};

		\node[fill,circle,inner sep = 0.05 cm] (x) at (-1,0) {};
		\node[below = 0.1 cm of x] () {$x$};
      	
        \node[fill,circle,inner sep = 0.05 cm] (y) at (1,0) {};
		\node[below = 0.1 cm of y] () {$y$};

		\node[fill,circle,inner sep = 0.05 cm] (a1) at (-3,0.75) {};
		\node[fill,circle,inner sep = 0.05 cm] (a2) at (-3,-0.75) {};
		\node at (-3,0.25) {\textbf{.}};
		\node at (-3,0) {\textbf{.}};
		\node at (-3,-0.25) {\textbf{.}};
		\draw (a1) to (arbre) to (a2);

		\draw[decorate,decoration={brace,raise=0.1cm}] (-3.2,1) to (-0.8,1);
		\node at (-2,1.4) {$l$ vertices};

		\node[fill,circle,inner sep = 0.05 cm] (chaine1) at (2,0) {};
		\draw (chaine1) to (y);

		\node (c1) at (2.75,0) {\textbf{.}};
		\node (c2) at (3,0) {\textbf{.}};
		\node (c3) at (3.25,0) {\textbf{.}};
		\node[fill,circle,inner sep = 0.05 cm] (c4) at (4,0) {};

		\draw (r) to (c1);
		\draw (c3) to (c4);

		\node[fill,circle,inner sep = 0.05 cm] (chaine2) at (2,1) {};
		\draw (chaine2) to (y);

		\node (c'1) at (2.75,1) {\textbf{.}};
		\node (c'2) at (3,1) {\textbf{.}};
		\node (c'3) at (3.25,1) {\textbf{.}};
		\node[fill,circle,inner sep = 0.05 cm] (c'4) at (4,1) {};

		\draw (chaine2) to (c'1);
		\draw (c'3) to (c'4);

		\node[fill,circle,inner sep = 0.05 cm] (chaine3) at (2,-1) {};
		\draw (chaine3) to (y);

		\node (c''1) at (2.75,-1) {\textbf{.}};
		\node (c''2) at (3,-1) {\textbf{.}};
		\node (c''3) at (3.25,-1) {\textbf{.}};
		\node[fill,circle,inner sep = 0.05 cm] (c''4) at (4,-1) {};

		\draw (chaine3) to (c''1);
		\draw (c''3) to (c''4);

		\node () at (2,0.5) {\textbf{.}};
		\node () at (2,0.65) {\textbf{.}};
		\node () at (2,0.35) {\textbf{.}};
		\node () at (2,-0.5) {\textbf{.}};
		\node () at (2,-0.65) {\textbf{.}};
		\node () at (2,-0.35) {\textbf{.}};

		\draw[decorate,decoration={brace,raise=0.1cm}] (0.8,1.2) to (4.2,1.2);
		\node at (2.5,1.6) {$m$ vertices};
        
        \draw[decorate,decoration={brace,raise=0.1cm}] (4.6,1.2) to (4.6,-1.2);
		\node at (5.45,0) {$k$ chains};
	\end{tikzpicture}
	\end{center}
    \caption{Graph $W_{k,l,m}$}\label{pincer}
\end{figure}

\begin{proposition}\label{prop:upperbound} 
For all $k\geq 2, \frac 12\leq \alpha_{I,k} \leq \frac{1}{\varphi}$, more precisely:\\
(i) $\alpha_I= \frac 12$, which means that the greedy algorithm is optimal for {\sc Firefighter} in finite trees;\\
(ii) $\alpha_{I,2} = \frac{1}{\varphi}$, which means that \Cref{algo:2firefighters} is optimal if at most two firefighters are available;\\
(iii) $\alpha_{I,4}<\frac{1}{\varphi}$.
\end{proposition}

\begin{proof} 
\Cref{prop:greed} shows that $\alpha_I\geq\frac 12$.
Given integers $l,m,k$ such that $k|m-1$, we define the graph $W_{k,l,m}$ as shown in \Cref{pincer}. We will assume that $m>k^2$.

{\em (i)} Let us consider an online algorithm  for $W_{k,l,m}$. As established in \Cref{sec:online-version} we can assume that the online algorithm plays in $T_i$ at turn $i$. If $f_1=1$, the algorithm will protect either $x$ or $y$. If $x$ is selected and the firefighter sequence is $(1,1,0,0,0,\ldots)$, our online algorithm protects the branch of $x$ and one of the $k$ chains, while the optimal offline algorithm protects $y$ and the star. Its performance is then $\frac{l+\frac{m-1}k}{l+m-1}$. If, however, $y$ is protected instead during the first turn and if the firefighter sequence is $(1,0,1,1,1,\ldots)$, the online algorithm protects the branch of $y$ and one vertex of the star whilst the optimal algorithm protects the branch of $x$ as well as the $k$ chains, minus $\frac{k(k+1)}2$ vertices. If $l=m-1=k^4$, for large values of $k$, the online  algorithm which protects $x$ is more performant and its competitive ratio is $\frac{1+\frac 1k}2$. Having $k\rightarrow +\infty$ shows that $\alpha\leq \frac 12$. Since the greedy algorithm $\mathbf{GR}$ guarantees $\alpha_I\geq\frac 12$, we have $\alpha_I=\frac 12$.

{\em (ii)} Consider the graphs $W_{1,l,\lfloor\varphi l\rfloor}$. If the online algorithm protects $x$, the adversary selects the sequence $(1,0,0,0,\ldots)$, whereas if the online algorithm protects $y$, $(1,0,1,0,0,0,\ldots)$ is selected. In both cases, the performance tends to $\frac 1\varphi$ when $l\rightarrow +\infty$.

{\em (iii)} If at most 4 firefighters are available, the graph $W_{4,901,1001}$ gives an example where $\frac 1\varphi$ cannot be reached. Indeed, if $f_1=1$ and the online algorithm protects $x$, then the adversary will select the sequence $(1,1,0,0,0,\ldots)$, as in the proof of {\em (i)}, for a performance of $\frac{1151}{1901}$. If the online algorithm protects $y$, since firefighters are limited to 4, the adversary will select $(1,0,1,1,1,0,0,0,\ldots)$, for a performance of $\frac{1002}{1645}$. This second choice is slightly better; however $\frac{1002}{1645}<\frac 1\varphi$.
\end{proof}

We have also proved that there is a $\frac{1}{\varphi}$-competitive algorithm if three firefighters are presented (i.e., $\alpha_{I,3}=\frac{1}{\varphi}$). This algorithm is similar to \Cref{algo:2firefighters} in that it places the first firefighter on one of the three largest branches and greedily places each of the other two on the largest branch available at the time they are presented. However, our proof involves a much more technical case-by-case analysis, and will not be detailed here. 

\section{Separating Firefighter Sequences}\label{sec:separe}

\subsection{Definitions}

We now consider the fractional firefighter problem on infinite graphs.  We say that a sequence of firefighters $(f_i)$ is \emph{weaker} than $(f'_i)$ (or $(f'_i)$ is \emph{stronger} than $(f_i)$) if $\forall k, S_k\leq S'_k=\sum\limits_{i=1}^k f'_i$, and we write $(f_i)\preceq (f'_i)$. If we also have $\exists k: S_k < S'_k$, $(f_i)$ is said to be \emph{strictly weaker} than $(f'_i)$ and we write $(f_i)\prec(f'_i)$. 

\begin{lemma}\label{weakness} If the fire can be contained in the instance $(G,r,(f_i))$ and if $(f_i)\preceq (f'_i)$, then the fire can also be contained in $(G,r,(f'_i))$ by an online algorithm that knows $(f_i)$ in advance.
\end{lemma}

\begin{proof}
Given a winning strategy in the instance $(G,r,(f_i))$, if $(f'_i)$ firefighters are available, we contain the fire by protecting the same vertices, possibly earlier than in the initial strategy.
\end{proof}

However, if $(f_i)\prec(f'_i)$, for {\sc Fractional Firefighter}, it is not always the case that there is an infinite graph $G$ such that the fire can be contained in $(G,r,(f'_i))$ but not in $(G,r,(f_i))$ (see \Cref{ex1}). We call such a $G$ a \emph{separating graph} for $(f_i)$ and $(f'_i)$, and we say that $G$ \emph{separates} $(f_i)$ and $(f'_i)$ in $N$ turns if the fire can be contained in $N$ turns for $(f'_i)$ but not for $(f_i)$. In this section, we give sufficient conditions for the existence of a separating graph.

\begin{example}\label{ex1}
Let $f_1=1$, $f'_1=1.5$ and $\forall i\geq 2$, $f_i=f'_i=0$. Although $(f_i)\prec(f'_i)$, no graph separates those two sequences.
\end{example}

Note that for {\sc Firefighter}, the problem is trivial, as shown in \Cref{integral separation}.

\subsection{Spherically Symmetric Trees}\label{SST}
Given a sequence $(a_i)\in (\mathbb N^*)^{\mathbb N^*}$, the \emph {spherically symmetric tree} $T((a_i))$ is the tree rooted in $r$ where every vertex of level $i-1$ has $a_i$ children \cite{lyons_peres_2017}. Note that if $T=T((a_i))$, we have $|T_i|=\prod_{j=1}^i a_j$. The total amount of fire at level $i$ is the sum of the amounts of fire on all vertices of level $i$.

\begin{proposition}\label{spherically symmetric trees}
In the instance $(T((a_i)),r,(f_i))$, the total amount of fire that spreads to level $i$ is $max\{0,F_i\}$, where $F_0=1$ and $F_i= a_i F_{i-1}-f_i$ for all $i$.
\end{proposition}

\begin{proof}
At turn $i$, the player only protects vertices on level $i$. If no protection is placed, the total amount of fire is multiplied by $a_i$. Hence, if $f_i\geq a_iF_{i-1}$, the fire is contained; else, the total amount of fire spreading to level $i$ is $a_iF_{i-1}-f_i$, regardless of how the protection is distributed among the vertices of level $i$.
\end{proof}

\begin{corollary}\label{integral separation}
Let $(f_i)$ and $(f'_i)$ be two distinct integral valued sequences. There is a spherically symmetrical tree which separates $(f_i)$ and $(f'_i)$.
\end{corollary}
\begin{proof}
Let $k$ be the first rank where $f_k\neq f'_k$. We may assume that $f_k<f'_k$. It follows from \Cref{spherically symmetric trees} that in the instance $(T((f_i+1)),r,(f_i))$, the amount of fire that spreads to each level is equal to 1. Yet, in $(T((f_i+1)),r,(f'_i))$, the fire is contained at turn $k$.
\end{proof}

For the purpose of the following technical lemma, we define a new firefighter sequence, which may include a negative term.  
Given a firefighter sequence $(f_i)$ and two non-zero integers $k$ and $\epsilon$, we define the firefighter sequence $(f^{(k,\epsilon}_i)$ via:
$$
f_i^{(k,\epsilon)}=\left\{ 
\begin{array}{lll}
f_k+\epsilon &\;\; & \text{if} \;\; i=k\\
f_{k+1}-\epsilon &\;\; & \text{if} \;\; i=k+1\\
f_i &\;\; & \text{otherwise}
\end{array}
\right.
$$

Note that there is a possibility that $f_{k+1}^{(k,\epsilon)}$ might be negative; however, this does not impact the reasoning. We also define the sequence $(F_i^{(k,\epsilon)})$ via $F_0^{(k,\epsilon)}=1$ and $F_i^{(k,\epsilon)}=a_iF_{i-1}^{(k,\epsilon)}-f_i^{(k,\epsilon)}$. It follows from \Cref{spherically symmetric trees} that the amount of fire which spreads to level $i$ in the instance ($T((a_i)),r,(f_i^{(k,\epsilon)}))$ is $\max\{0,F_i^{(k,\epsilon)}\}$.

\begin{lemma}\label{fire in AB}
The spherically symmetric tree $T=T((a_i))$ separates $(f_i)$ and $(f^{(k,\epsilon)}_i)$ if and only if there is a rank $N$ such that: $A\leq \sum_{i=k+2}^N \dfrac{f_i}{\prod_{j=k+2}^i a_j}<B$, where $A=F^{(k,\epsilon)}_{k+1}$ and $B=F_{k+1}$.
\end{lemma}

\begin{proof}
It follows from \Cref{spherically symmetric trees} that $F_n=\prod_{j=1}^n a_j -\sum_{i=1}^n f_i\prod_{j=i+1}^n a_j$. So, $F_n=|T_n|(1-\sum_{i=1}^n\frac{f_i}{|T_i|})$. The condition for $T((a_i))$ to separate $(f_i)$ and $(f^{(k,\epsilon)})$ can be stated as follows: there is a rank $N$ such that $F^{(k,\epsilon)}_N\leq 0<F_N$. Hence, there is an $N$ such that
$$|T_N|(1-\sum_{i=1}^N\frac{f^{(k,\epsilon)}_i}{|T_i|})\leq 0<|T_N|(1-\sum_{i=1}^N\frac{f_i}{|T_i|}).$$ Therefore $$1-\sum_{i=1}^{k+1}\frac{f^{(k,\epsilon)}_i}{|T_i|}\leq\sum_{i=k+2}^N\frac{f_i}{|T_i|}<1-\sum_{i=1}^{k+1}\frac{f_i}{|T_i|}.$$ And finally, $$A\leq \sum_{i=k+2}^N \dfrac{f_i}{\prod_{j=k+2}^i a_j}<B,$$
with $A=F^{(k,\epsilon)}_{k+1}=|T_{k+1}|(1-\sum_{i=1}^{k+1}\frac{f^{(k,\epsilon)}_i}{|T_i|})$ and $B=F_{k+1}=|T_{k+1}|(1-\sum_{i=1}^{k+1}\frac{f_i}{|T_i|})$.
\end{proof}

\begin{proposition}\label{prop:Fire in AB}
Given two sequences $(f_i)$ and $(f'_i)$ such that $(f_i)\preceq (f'_i)$, let $k$ be the smallest integer such that $f_k\neq f'_k$ and let $\epsilon=f'_k-f_k$. The spherically symmetric tree $T=T((a_i))$ separates $(f_i)$ and $(f'_i)$ if and only if there is a rank $N$ such that: $A\leq \sum_{i=k+2}^N \dfrac{f_i}{\prod_{j=k+2}^i a_j}<B$, \\where $A=F^{(k,\epsilon)}_{k+1}$ and $B=F_{k+1}$.
\end{proposition}

\begin{proof}  
We have $(f^{(k,\epsilon)}_i)\preceq (f'_i)$, as indeed, $(f^{(k,\epsilon)}_i)$ is the weakest sequence in $\{(g_i)\in\mathbb R^\mathbb N|\forall i<k$, $g_i=f_i$,  $g_k=f_k+\epsilon$ and $(f_i)\preceq(g_i)\}$. Hence, any tree separating $(f_i)$ and $(f^{(k,\epsilon)}_i)$ also separates $(f_i)$ and $(f'_i)$. We conclude using \Cref{fire in AB}.
\end{proof}

\subsection{Targeting Game}
	Given the form of the condition in \Cref{prop:Fire in AB}, we can view this as a special case of a purely numerical problem, which we will call the \emph{targeting game}.
	The instance of the problem is given by two positive real numbers, $A<B$, and a sequence of non-negative real numbers $(f_i)$ which represents the movements towards the target $[A,B[$.
	The player starts at position $u_0=0$ with an initial step size of $1$. We denote by $\delta_i$ the step size at turn $i$, so $\delta_0=1$.
	At each turn $i>0$, the player chooses a positive integer $a_i$ by which he will divide the previous step size, that is to say $\delta_i=\dfrac{\delta_{i-1}}{a_i}=\prod_{j=1}^i a_j^{-1}$. Then, the position of the player is updated with the rule $u_i=u_{i-1}+f_i\delta_i$.
	If there is an integer $N$ such that $u_N\in[A,B[$, then the player wins with the strategy $(a_i)$.

The targeting game can be summarised as follows:
Given $0<A<B$ and a sequence $(f_i)$, is there an $N$ and a sequence $(a_i)$ such that $A\leq \sum_{i=1}^N \dfrac{f_i}{\prod_{j=1}^i a_j} <B$~?

We give two sufficient conditions on the data to ensure the existence of a winning strategy for the player.

\begin{theorem}\label{divergent}
	If there is an $N$ such that $\sum_{i=1}^N f_i \geq A\left\lceil\dfrac{A}{B-A}\right\rceil$, then there exists a sequence $(a_i)$ with $a_i=1, \forall i\geq2$ such that the player wins the targeting game at turn $N$ by selecting $(a_i)$.
\end{theorem}

\begin{proof}
First, note that if $m\geq \dfrac{A}{B-A}$, then $(m+1)A\leq mB$. It follows that $$[A\left\lceil\dfrac A{B-A}\right\rceil,+\infty[ \subset \bigcup_{k\in\mathbb{N}^*} [kA,kB[.$$ Hence, there is a $k$ such that $kA\leq \sum_{i=1}^N f_i<kB$. So, $A\leq \sum_{i=1}^N \dfrac{f_i}{k}<B$. Therefore, if the player chooses $a_1=k$ and $a_i=1$, for $i\geq 2$, he will have reached the target at turn $N$.
\end{proof}

\begin{theorem}\label{gluttonous rabbit}
	If $|\{i:f_i\geq B\}|\geq \log_2\left(\dfrac{B}{B-A}\right)$, then the player wins the targeting game by choosing at each turn $i$ the smallest positive integer $a_i$ such that $u_i<B$.
\end{theorem}
\begin{proof}
	Consider a turn $i$ such that $a_i>1$. Given that the player chooses the minimum $a_i$, it follows that $B\leq u_{i-1}+\dfrac{\delta_{i-1}}{a_i-1}f_i$. By definition of $u_i$ and $\delta_i$, we have $\delta_{i-1}f_i=a_i(u_i-u_{i-1})$, so $B\leq u_{i-1}+\dfrac{a_i}{a_i-1}(u_i-u_{i-1})$. Then $B(a_i-1)\leq u_{i-1}(a_i-1)+a_i(u_i-u_{i-1})$ and
\begin{equation}
	\label{equation_Bui}
	a_i(B-u_i)\leq B-u_{i-1}.
\end{equation}

Now consider the sequence $x_i=\dfrac{B-u_i}{\delta_i}$. By dividing \Cref{equation_Bui} by $\delta_{i-1}$, we see that $x_i\leq x_{i-1}$ when $a_i>1$. When $a_i=1$, we also have $x_i=x_{i-1}-f_i\leq x_{i-1}$. 
Thus $(x_i)$ is non-increasing, and $\forall i, x_i\leq x_0=B$.

At any turn $i$ where $f_i\geq B$, we have $f_i\geq x_{i-1}$ and $$\delta_{i-1}f_i\geq \delta_{i-1}x_{i-1}=B-u_{i-1}>u_i-u_{i-1}=\delta_if_i.$$ So $\delta_{i-1}>\delta_i$, and $a_i>1$.  
It then follows from \Cref{equation_Bui} that $B-u_i\leq \dfrac{B-u_{i-1}}{2}$.

Note also that, since $(x_i)$ and $(\delta_i)$ are non-increasing, $(B-u_i)$ is also non-increasing. 
Hence, for all $N$, we have: $$B-u_N\leq \frac{B-u_0}{2^{|\{i\leq N: f_i\geq B\}|}}=\frac{B}{2^{|\{i\leq N: f_i\geq B\}|}}.$$

Finally, choosing $N$ such that $|\{i\leq N:f_i\geq B\}|\geq \log_2\left(\dfrac{B}{B-A}\right)$, we have $A\leq u_N<B$.
\end{proof}

\begin{proposition}\label{application}
Given $(f_i)<(f'_i)$, let $k$ be the smallest integer such  that $f_k\neq f'_k$ and let $\epsilon=f'_k-f_k$. If there is an $N$ such that $\sum_{k+2}^N f_i\geq 2\left\lceil \dfrac 2\epsilon\right\rceil$ or $|\{k+2\leq i\leq N;f_i\geq 2\}|> 1-log_2\epsilon$, then there is a spherically symmetric tree which separates $(f_i)$ and $(f'_i)$ in $N$ turns.
\end{proposition}

\begin{proof}

For $i\leq k$, we choose the smallest $a_i$ such that $F_i>0$; i.e. $a_i=\left\lfloor\frac{f_i}{F_{i-1}}\right\rfloor+1$. We then choose $a_{k+1}=\max\{2,\left\lfloor\frac{f_{k+1}}{F_{k}}\right\rfloor+1\}$.

Using \Cref{prop:Fire in AB}, it is sufficient to have a rank $N$ such that: 
$$A\leq \sum_{i=k+2}^N \dfrac{f_i}{\prod_{j=k+2}^i a_j}<B,$$ 
where $A=|T_{k+1}|(1-\sum_{i=1}^{k+1} \frac{f^{(k,\epsilon)}_i}{|T_i|})$ and $B=|T_{k+1}|(1-\sum_{i=1}^{k+1} \frac{f_i}{|T_i|})$

\bigskip
It follows from the choice of $a_i$ that for $i\leq k$, $(a_i-1)F_{i-1}-f_i\leq 0$. Hence, $a_iF_{i-1}-f_i\leq F_{i-1}$, and $F_i\leq F_{i-1}$. If $a_{k+1}=\left\lfloor\frac{f_{k+1}}{F_{k}}\right\rfloor+1$, then $F_{k+1}\leq F_k$. Otherwise, if $a_{k+1}=2$, $F_{k+1}\leq 2F_k$. Finally, we have $B=F_{k+1}\leq 2F_0=2$.

Also, $B-A= \sum_{i=1}^{k+1} (f^{(k,\epsilon)}_i-f_i)\prod_{j=i+1}^{k+1}a_j=(f^{(k,\epsilon)}_k-f_k)a_{k+1}+(f^{(k,\epsilon)}_{k+1}-f_{k+1})=(a_{k+1}-1)\epsilon$. Having chosen $a_{k+1}\geq 2$, we have $B-A\geq\epsilon$.

Thus, we have an $N$ such that $\sum_{k+2}^N f_i\geq 2\left\lceil \dfrac 2\epsilon\right\rceil\geq  A\left\lceil\dfrac{A}{B-A}\right\rceil$ or $|\{k+2\leq i\leq N|f_i\geq 2\}|> 1-log_2\epsilon\geq \log_2\left(\dfrac{B}{B-A}\right)$. The result follows by applying \Cref{divergent} or \Cref{gluttonous rabbit}.

\end{proof}

\begin{remark}
In the case where  $|\{k+2\leq i\leq N;f_i\geq 2\}|> 1-log_2\epsilon$, the sequence $(a_i)$ is entirely created by a greedy algorithm which selects the minimum value of $a_i$ such that $F_i>0$ (and $a_{k+1}\geq 2$). The value of $a_i$ is therefore a function of $F_{i-1}$ and $f_i$.
\end{remark}

\section{Firefighting Sequence vs. Level Growth}\label{sec:asymptotic}

\subsection{Infinite Offline Instances}

On infinite trees, the objective is to contain the fire. We will consider only locally finite trees, i.e., trees where each vertex has finite degree. Given a locally finite rooted tree $T$ with at least one infinite branch, we consider $T^\ast$ the leafless sub-tree obtained from $T$ by pruning finite branches. Formally, $T^\ast$ is the union of all leafless sub-trees of $T$ with the same root, where the union $T_1\cup T_2$ of two such sub-trees of $T$ is the sub-tree induced by $V(T_1)\cup V(T_2)$. Since $T$ is locally finite, the fire is contained on $T$ if and only if it is contained on $T^\ast$.  %Hence, finite branches of infinite trees are irrelevant. In particular, the problem is trivial on infinite trees without infinite branches and we exclude this case. Given a rooted tree $T$ with at least one infinite branch, let us consider the union $T^\ast$ of all leafless sub-trees of $T$ with the same root, where the union $T_1\cup T_2$ of two such sub-trees of $T$ is the sub-tree induced by $V(T_1)\cup V(T_2)$. $T^\ast$ is leafless (it is obtained from $T$ by pruning finite branches).   Playing on $T$   
%is equivalent  to playing on $T^\ast$.
Hence, without loss of generality, we may restrict the infinite case to leafless trees. Note that if $T$ is leafless, then $(|T_i|)$ is non-decreasing.

%Intuitively, it seems that when the fire cannot be contained on an infinite tree, it means that the number of vertices per level must grow faster, in some sense, than the firefighter sequence. 
Intuitively, it seems that when the firefighter sequence grows faster, in some sense, than the number of vertices per level, the firefighter should be able to contain the fire. Following this line of reasoning, \Cref{reviewer} and \Cref{comparison} give criteria for infinite instances to be winning based on the asymptotic behaviours of those two sequences.

\begin{proposition}\label{reviewer}
Let $(T,r,(f_i))$ be an instance of {\sc Fractional Firefighter} where $T$ is a  tree of infinite height. If $\sum_{i=1}^{+\infty}\frac{f_i}{|T_i|}>1$, then the instance is winning.
\end{proposition}

\begin{proof}
The firefighter wins by spreading at each turn $n$ the amount of protection evenly among all vertices of $T_n$. The amount of fire that reaches $v\in T_n$  is $\max\{0,1-\sum_{i\leq n}\frac{f_i}{|T_i|}\}$. Hence, the fire is contained after a finite number of turns.
\end{proof}

Unfortunately, we need a more complex criterion to obtain a sufficient condition to win in both {\sc Firefighter} and {\sc Fractional Firefighter}.

\begin{theorem}\label{comparison}
Let $(T,r,(f_i))$ be an instance of {\sc Firefighter} or {\sc Fractional Firefighter} where $T$ is a leafless tree. If $S_i\rightarrow +\infty$ and $\frac{S_i}{|T_i|} \nrightarrow 0$, then the instance $(T,r,(f_i))$ is winning for the firefighter.
\end{theorem}

The proof of \Cref{comparison} will require the following lemma, which is probably well-known:

\begin{lemma}\label{taupin}
If $(u_n)$ is a positive sequence that increases towards $+\infty$, then  $\sum \frac{u_n-u_{n-1}}{u_n}$ diverges.
\end{lemma}

\begin{proof}
Let $v_n=\frac{u_n-u_{n-1}}{u_n}$. If $v_n\nrightarrow 0$, then $\sum v_n$ diverges. Let us assume that $v_n\rightarrow 0$. Since $\frac{u_{n-1}}{u_n}=1-v_n$, we have $ln\frac{u_0}{u_n}=ln\prod_{i=1}^n 1-v_i=\sum_{i=1}^n ln(1-v_i)$. Since $u_n\rightarrow +\infty$, $ln\frac{u_0}{u_n}\rightarrow -\infty$, so $\sum ln(1-v_i)\rightarrow -\infty$, and since $v_n\rightarrow 0$, $\sum_{i=1}^n v_i\rightarrow +\infty$.
\end{proof}
We may now prove \Cref{comparison}.

\begin{proof}
Since $T$ is leafless, $(|T_i|)$ is non-decreasing  and
since $\frac{S_i}{|T_i|} \nrightarrow 0$, there is a positive constant $C$ and an increasing injection $\sigma:\mathbb N\rightarrow \mathbb N$ such that $$\forall i, |T_{\sigma(i)}|\leq C S_{\sigma(i)}<\infty.$$

Let $a:V(T)\rightarrow [0,1]$ denote the amount of protection we will place on each vertex. In order to describe the amount of each vertex that remains unprotected at the end of each turn, we use a sequence of labellings $l_i:V(T)\rightarrow [0,1]$. Initially, all vertices are unprotected, so $l_0=1$. At turn $i$, protection is placed on vertices of $T_i$, so $\forall v\in V(T), l_i(v)=l_{i-1}(v)-\sum_{v'\in T_i} a(v')\mathds 1_{v'\unlhd v}$. 
For any $W\subset V(T)$ and any labelling $l$, we define $l(W)=\sum_{v\in W} l(v)$.

For all $i$ and $h\in \mathbb N^*$ with $h>i$, for all $v\in T_i$, let $w_h(v)=|\{v'\in T_h,v\lhd v'\}|$. Thus, $\sum_{v\in T_i} w_h(v)=|T_h|$ and for all $j<i$,
\begin{equation}\label{eq1a}
 \sum_{v\in T_i} w_h(v)l_j(v)=l_j(T_h).
\end{equation}
It follows from \Cref{taupin} that $\sum\frac{S_{\sigma(i)}-S_{\sigma(i-1)}}{S_{\sigma(i)}}$ diverges.  Hence $\prod(1+\frac{S_{\sigma(i)}-S_{\sigma(i-1)}}{CS_{\sigma(i)}})$ also diverges. Let $N$ be such that $\prod_{i=1}^N(1+\frac{S_{\sigma(i)}-S_{\sigma(i-1)}}{CS_{\sigma(i)}})>2C$ and let $h$ be such that $S_{\sigma(h)}>2S_{\sigma(N)}$.

We consider the following strategy. At each turn $i$,  
we protect the vertices which have the most descendants in level $\sigma(h)$, i.e., $a(v),v\in T_i$ is an optimal solution of the following linear program:
$$
\left\{ \begin{array}{lll}
\max \sum_{v\in T_i} a(v)w_{\sigma(h)}(v)\\
 a(v)\leq l_{i-1}(v) \; ; \; v\in T_i\\
\sum_{v\in T_i}a(v)\leq f_i
\end{array}\right.$$

If $f_i\geq l_{i-1}(T_i)$ then we can protect the whole level $T_i$, thus $T_{\sigma(h)}$, and the fire is contained. So we assume $f_i< l_{i-1}(T_i)$ for all $i<\sigma(h)$. Then, $f_i\frac{l_{i-1}(v)}{l_{i-1}(T_i)},v\in T_i$ is a solution of the linear program with $\sum_{v\in T_i}f_i\frac{l_{i-1}(v)}{l_{i-1}(T_i)}= f_i$. It follows, by optimality and using \Cref{eq1a}, that:
$$\sum_{v\in T_i} a(v)w_{\sigma(h)}(v)\geq \sum_{v\in T_i} f_i\frac{l_{i-1}(v)}{l_{i-1}(T_i)}w_{\sigma(h)}(v)=\frac{f_il_{i-1}(T_{\sigma(h)})}{l_{i-1}(T_i)}$$
Note that for $j\leq i$, $l_{j-1}(T_j)\leq |T_j|\leq |T_{\sigma(i)}|$.
Hence, 
\begin{eqnarray*}
l_{\sigma(i-1)}(T_{\sigma(h)})-l_{\sigma(i)}(T_{\sigma(h)})&=&\sum_{j=\sigma(i-1)+1}^{\sigma(i)}\sum_{v\in T_j} a(v)w_{\sigma(h)}(v)\\
&\geq& \sum_{j=\sigma(i-1)+1}^{\sigma(i)}\frac{f_jl_{j-1}(T_{\sigma(h)})}{l_{j-1}(T_j)}\\
&\geq& \sum_{j=\sigma(i-1)+1}^{\sigma(i)}\frac{f_jl_{\sigma(i)}(T_{\sigma(h)})}{|T_{\sigma(i)}|}\\
&\geq& \frac{S_{\sigma(i)}-S_{\sigma(i-1)}}{CS_{\sigma(i)}}l_{\sigma(i)}(T_{\sigma(h)})\ \ ({\rm since\  } |T_{\sigma(i)}|\leq C S_{\sigma(i)}).
\end{eqnarray*}

So $$l_{\sigma(i-1)}(T_{\sigma(h)})\geq (1+\frac{S_{\sigma(i)}-S_{\sigma(i-1)}}{CS_{\sigma(i)}})l_{\sigma(i)}(T_{\sigma(h)}).$$ 

Therefore, 
$$|T_{\sigma(h)}|\geq l_{\sigma(0)}(T_{\sigma(h)})\geq \prod_{i=1}^N(1+\frac{S_{\sigma(i)}-S_{\sigma(i-1)}}{CS_{\sigma(i)}}) l_{\sigma(N)}(T_{\sigma(h)})>2Cl_{\sigma(N)}(T_{\sigma(h)}).$$ 
And consequently,
$$l_{\sigma(N)}(T_{\sigma(h)})\leq \frac {|T_{\sigma(h)}|}{2C}\leq \frac 12S_{\sigma(h)}\leq S_{\sigma(h)}-S_{\sigma(N)}.$$

This means that the firefighters available between turns $\sigma(N)$ and $\sigma(h)$ outnumber the unprotected  vertices on level $\sigma(h)$. Hence, the strategy will win in at most $\sigma(h)$ turns.
\end{proof}

Conversely, asymptotic behaviours cannot guarantee that an instance will be losing. Indeed, if $f_1\geq |T_1|$, the instance is winning regardless of asymptotic behaviours. However, having selected asymptotic behaviours where the levels of the tree grow faster than the firefighter sequence, \Cref{losing spherically symmetric} guarantees that some instances with those asymptotic behaviours will be losing.

\begin{theorem}\label{losing spherically symmetric} Let $(t_i)\in {\mathbb N^*}^{\mathbb N^*}$ and $(f_i)\in {\mathbb R^+}^{\mathbb N^*}$ be such that $(t_i)$ is non-decreasing and tends towards $+\infty$. Then, $\sum\frac{f_i}{t_i}$ converges if and only if there exists a spherically symmetric tree $T$ rooted in $r$ such that:
\begin{itemize}
\item $\exists N:\forall i\geq N, \frac{t_i}2\leq |T_i|\leq t_i$
\item the instance $(T,r,(f_i))$ is losing for {\sc (Fractional) Firefighter}.
\end{itemize}
\end{theorem}

\begin{proof}
1) Suppose that $\sum\frac{f_i}{t_i}$ converges. Let $M$ be such that $\sum_{i=M+1}^{+\infty} \frac{f_i}{t_i}<\frac 14$ and let $N>M$ be such that $t_N>4S_M$. We choose $a_1=t_N$, $a_i=1$ for $2\leq i\leq N$, and $a_i=\left\lfloor\frac{t_i}{\prod_{j=1}^{i-1} a_j}\right\rfloor$ for $i>N$. We will show that $T=T((a_i))$ is a solution.

Let us show by induction that $\forall i\geq N, \frac{t_i}2\leq |T_i|\leq t_i$. Note that  $$|T_N|=\prod_{j=1}^N a_j=t_N.$$ 
Assume that the result holds for $i-1$ where $i>N$: $$\frac{t_{i-1}}2\leq\prod_{j=1}^{i-1} a_j \leq t_{i-1}.$$
Since $t_i\geq t_{i-1}$, we have $a_i\geq 1$. Hence, $$a_i\leq \frac{t_i}{\prod_{j=1}^{i-1} a_j}\leq a_i+1\leq 2a_i.$$ 
So, $\frac{t_i}2\leq |T_i|\leq t_i$, and the result holds for all $i\geq N$.

\bigskip
Since $T$ is spherically symmetric, the amount of fire that spreads to level $n$ is $\max(0,F_n)$, where $F_n=|T_n|(1-\sum_{i=1}^n\frac{f_i}{|T_i|})$. For $n>N$, we have:
\begin{eqnarray*}
\sum_{i=1}^n\frac{f_i}{|T_i|}&=&\frac {S_N}{t_N}+ \sum_{i=N+1}^n \frac{f_i}{|T_i|}\\
	&\leq & \frac {S_M}{t_N} +\frac 1{t_N}\sum_{i=M+1}^N f_i +2\sum_{i=N+1}^n \frac{f_i}{t_i}\\
	&<&\frac 14+\frac 1{4}+\frac 24= 1.
\end{eqnarray*}

Hence, $F_n>0$ for all $n$, and therefore the fire cannot be contained.

\bigskip
\noindent 2) Conversely, if $\sum\frac{f_i}{t_i}$ diverges and $T=T((a_i))$ is such that $\exists N:\forall i\geq N,\\ \frac{t_i}2\leq |T_i|\leq t_i$, then $\sum\frac{f_i}{|T_i|}$ also diverges. It follows that $F_n=|T_n|(1-\sum_{i=1}^n\frac{f_i}{|T_i|})$ is negative above a certain rank. Hence, the fire is contained.
\end{proof}

\begin{corollary}
Let $(t_i)\in {\mathbb N^*}^{\mathbb N^*}$ and $(f_i)\in {\mathbb R^+}^{\mathbb N^*}$ be such that $(t_i)$ is non-decreasing and tends towards $+\infty$. Let $S_i=\sum_{1\leq k\leq i} f_k$. If $S_i\rightarrow +\infty$ and $\frac{S_i}{t_i} \nrightarrow 0$, then $\sum\frac{f_i}{t_i}$ diverges.
\end{corollary}

\begin{proof}
If $\sum\frac{f_i}{t_i}$ were convergent, it follows from \Cref{losing spherically symmetric} that there would be a spherically symmetric tree $T$ such that:
\begin{itemize}
\item $\exists N:\forall i\geq N, \frac{t_i}2\leq |T_i|\leq t_i$
\item the instance $(T,r,(f_i))$ is losing for {\sc (Fractional) Firefighter}.
\end{itemize}

It then follows from \Cref{comparison} that $S_i\nrightarrow +\infty$ or $\frac{S_i}{|T_i|}\rightarrow 0$. Hence $S_i\nrightarrow +\infty$ or $\frac{S_i}{t_i}\rightarrow 0$. 
\end{proof}

It follows that for {\sc Fractional Firefighter}, \Cref{comparison} is weaker than \Cref{reviewer}. \Cref{comparison} remains interesting for {\sc Firefighter} and it  gives an alternative winning method for {\sc Fractional Firefighter}. 

\begin{remark}
Under the hypotheses of \Cref{losing spherically symmetric}, if $\sum\frac{f_i}{t_i}$ converges, we can create a losing instance $(T',r',(f_i))$ with $\forall i, |T_i|=t_i$ by adding $t_i-|T_i|$ leaves to level $i$ for all $i$. We will have $|T_n|=t_n$ without adding leaves if and only if there exists a spherically symmetric tree with $t_i$ vertices on level $i$ for all $i$.
\end{remark}

\begin{remark}
Remark 1.12 in~\cite{Dyer-containment} gives a sufficient condition for an instance to be losing for {\sc Firefighter} in a general graph satisfying some growth condition, using a similar criteria to the convergence of $\sum\frac{f_i}{t_i}$. In general, both results cannot be compared. In our set-up however, their result can be seen as the particular case of {\sc Firefighter} where $(t_i)=(\lambda^i)$   for some $\lambda$.   
\end{remark}

\subsection{Online Firefighting on Trees with Linear Level Growth}\label{sec:linear}

In the previous section, \Cref{reviewer} gives a winning strategy for online {\sc Fractional Firefighter} in cases where $\sum\frac{f_i}{|T_i|}>1$. However, \Cref{comparison} is limited to the offline case, as the winning strategy requires the player to be able to compute $\sigma(h)$ from the start. In this section, we give a result which works for online {\sc Firefighter} in the case of rooted trees $(T,r)$ where the number of vertices per level increases linearly, i.e. $|T_i|=\mathcal O(i)$. We say that such a tree has \emph{linear level growth}.

\begin{remark}
The linear level growth property of $T$ remains if we choose a different root $r'$. Indeed, if $d$ is the distance between $r$ and $r'$, the set of vertices at distance $i$ from $r'$ is included in $\bigcup_{j=i-d}^{i+d} T_j$, the cardinal of which is $\mathcal O(i)$.
\end{remark}

\begin{theorem}\label{linear level growth}

%There is an online algorithm for instances $(T,r,(f_i))$ of {\sc Firefighter} where $T$ has linear level growth, such that if some non-zero periodic sequence is weaker than $(f_i)$, the fire will be contained.

Let $\mathcal I$ be the set of instances $(T,r,(f_i))$ of {\sc Firefighter} where $T$ has linear level growth and there exists a non-zero periodic sequence which is weaker than $(f_i)$. There is an online algorithm which contains the fire for every instance in $\mathcal I$.

\end{theorem}
The proof of \Cref{linear level growth} will use the following lemma:

\begin{lemma}\label{product limit}
For any real number $a>1$, $\lim_{n\rightarrow +\infty}\prod_{j=1}^n \frac{ja-1}{ja}=0$. 
\end{lemma}

\begin{proof}
We have $\ln \prod_{j=1}^n \frac{ja-1}{ja}=\sum_{j=1}^n \ln(1-\frac 1{ja})$ and since $\sum_{j=1}^n \frac 1{ja}\rightarrow +\infty$, we have \\ $\sum_{j=1}^n \ln(1-\frac 1{ja})\rightarrow -\infty$. Hence, $\prod_{j=1}^n \frac{ja-1}{ja}\rightarrow 0$.
\end{proof}
We can now prove \Cref{linear level growth}:

\begin{proof} Since $T$ has linear level growth, let $C$ be such that $\forall i, |T_i|\leq Ci$. Without loss of generality, we assume $C>1$. That a non-zero periodic sequence is weaker than $(f_i)$ means that $(\mathds 1_{n|i})\preceq (f_i)$ for all $n$ greater than some $m$. First, we will give an offline strategy to contain the fire with one firefighter every $n$ turns. Then, we will show that online instances with $(\mathds 1_{n|i})\preceq(f_i)$ for an $n$ known to the player are winning. Finally, we will describe an online winning strategy when such a $(\mathds 1_{n|i})$ is unknown.

Given an integer $n$, let us first consider the instance $(T,r,(\mathds 1_{n|i}))$. It follows from \Cref{product limit} that there exists an integer $N$ such that $\prod_{j=1}^N \frac{Cnj-1}{Cnj}<\frac 1{2Cn}$. Let $h(n)=2nN$. A winning strategy for this offline instance is obtained by protecting at turn $nj$ the unprotected vertex of $T_{nj}$ with the highest number of descendants in level $h(n)$. Since $|T_{nj}|\leq Cnj$, the remaining number of unprotected vertices in $T_{h(n)}$ is reduced by at least $\frac 1{Cnj}$ of its previous value. So the number of unprotected vertices of $T_{h(n)}$ remaining after $nN$ turns is less than $|T_{h(n)}|\prod_{i=1}^N \frac {Cnj-1}{Cnj} \leq \frac{|T_{h(n)}|}{2Cn}\leq N$. Since $N$ firefighters remain to be placed between turns $N$ and $h(n)$, the strategy is winning in at most $h(n)$ turns.

If the player knows in advance that $(\mathds 1_{n|i})\preceq(f_i)$ for a given $n$, the above strategy can be adapted using \Cref{weakness}.

In the general case,  
assume that $(\mathds 1_{n|i})\preceq (f_i)$ for some $n$, but the player does not know which $n$. The online strategy proceeds as follows: we initially play as though under the assumption that $(\mathds 1_{n_0|i})\preceq(f_i)$ with $n_0=100$. If the fire is not contained by turn $h(n_0)$, or later on by turn $h(n_k)$, we choose $n_{k+1}=h(n_k)\left(\left\lceil S_{h(n_k)}\right\rceil+1\right)$. We now assume that $(\mathds 1_{h(n_k)|i})\preceq(f_i)$. It follows that after cancelling the first $h(n_k)$ terms of $(f_i)$, i.e., replacing $f_\ell$ with 0 for $\ell\leq h(n_k)$, the resulting sequence is stronger than $(\mathds 1_{n_{k+1}|i})$. So we can consider that the first $h(n_k)$ turns were wasted and follow the strategy for $n_{k+1}$ until turn $h(n_{k+1})$. Eventually, this strategy will win when $n_k$ is large enough.
\end{proof}
\section{Conclusion}\label{sec:conclusion}

The main thread of this paper is to consider a general sequence of number of firefighters available at each turn in {\sc (Fractional) Firefighter}, whereas most of the existing work on this topic considers constant sequences. We give first results, in the case of trees, for three independent research questions that arise when including such a sequence as part of the instance. 

We introduce the online version of {\sc (Fractional) Firefighter} on trees and provide initial results for the finite case. So far, our results outline the potential of this approach and suggest many open questions. To our knowledge, \Cref{prop:greed} is the first non-trivial competitive (and also approximation) analysis for {\sc Fractional Firefighter}
and a first question would be to investigate whether a better competitive ratio can be obtained  for {\sc Fractional Firefighter} in finite trees. Although the case of trees is already challenging, the main open question will be to study online  {\sc (Fractional) Firefighter problem} in other classes of finite graphs. 

As far as we know, the second question has never been considered yet. The existence of a separating tree for any two given firefighter sequences seems very hard in general. spherically symmetric trees  provide convenient examples of separating trees since they allow us to ignore the playing strategy. This allowed us to express the problem in terms of the targeting game, which completely hides the structure of the tree. An interesting question will be to investigate whether the existence of a separating tree implies that of a spherically symmetric separating tree. So far, we only considered the case where one of the sequences is weaker than the other. The general case remains fully open. 

We have shown that some conditions on the asymptotic behaviours of the firefighter sequence vs. the tree growth guarantee that the instance is winning. Yet, other conditions guarantee the existence of losing instances. We conjecture that all {\sc Firefighter} instances where $\sum \frac{f_i}{|T_i|}$ diverges are winning.

Finally, note that the question of approximating {\sc (Fractional) Firefighter} in finite trees for a general firefighter sequence is also an important research direction that, to our knowledge, remains uninvestigated.

\section*{Acknowledgements}
We are grateful to the anonymous reviewers for their helpful comments and suggestions, especially for highlighting~\Cref{reviewer}. We also acknowledge the support of GEO-SAFE,  H2020-MSCA-RISE-2015 project \# 691161.
\bibliographystyle{plain}
\bibliography{references}

\begin{thebibliography}{10}

\bibitem{ptas-tree}
David Adjiashvili, Andrea Baggio, and Rico Zenklusen.
\newblock Firefighting on trees beyond integrality gaps.
\newblock In {\em Proceedings of the Twenty-Eighth Annual {ACM-SIAM} Symposium
  on Discrete Algorithms, {SODA} 2017, Barcelona, Spain, Hotel Porta Fira,
  January 16-19}, pages 2364--2383, 2017.

\bibitem{albers}
Susanne Albers.
\newblock Online algorithms: a survey.
\newblock {\em Mathematical Programming Ser. B}, Ser. B(97):3--26, 2003.

\bibitem{Approx-algorithmica}
Elliot Anshelevich, Deeparnab Chakrabarty, Ameya Hate, and Chaitanya Swamy.
\newblock Approximability of the firefighter problem - computing cuts over
  time.
\newblock {\em Algorithmica}, 62(1-2):520--536, 2012.

\bibitem{paschos-online}
Giorgio Ausiello and Luca Becchetti.
\newblock On-line algorithms.
\newblock In V.~Th. Paschos, editor, {\em Paradigms of Combinatorial
  Optimization: Problems and New Approaches, Vol. 2}, chapter~15, pages
  473--509. ISTE - WILEY, London - Hoboken, 2010.

\bibitem{bazgan.chopin.fellow}
Cristina Bazgan, Morgan Chopin, Marek Cygan, Michael~R. Fellows, Fedor~V.
  Fomin, and Erik~Jan van Leeuwen.
\newblock Parameterized complexity of firefighting.
\newblock {\em Journal of Computer and System Sciences}, 80(7):1285--1297,
  2014.

\bibitem{bazgan-chopin-ries}
Cristina Bazgan, Morgan Chopin, and Bernard Ries.
\newblock The firefighter problem with more than one firefighter on trees.
\newblock {\em Discrete Applied Mathematics}, 161(7-8):899--908, 2013.

\bibitem{bonato12}
Anthony Bonato, Margaret-Ellen Messinger, and Paweł Pralat.
\newblock Fighting constrained fires in graphs.
\newblock {\em Theoretical Compututer Science}, 434:11--22, 2012.

\bibitem{Cai.Verbin.Yang}
Leizhen Cai, Elad Verbin, and Lin Yang.
\newblock Firefighting on trees: (1-1/e)-approximation, fixed parameter
  tractability and a subexponential algorithm.
\newblock In {\em Algorithms and Computation, 19th International Symposium,
  {ISAAC} 2008, Gold Coast, Australia, December 15-17, 2008. Proceedings},
  pages 258--269, 2008.

\bibitem{Chlebikova-tcs}
Janka Chleb{\'{\i}}kov{\'{a}} and Morgan Chopin.
\newblock The firefighter problem: Further steps in understanding its
  complexity.
\newblock {\em Theoretical Compututer Science}, 676:42--51, 2017.

\bibitem{Dyer-containment}
Danny Dyer, Eduardo Mart{\'{\i}}nez{-}Pedroza, and Brandon Thorne.
\newblock {The coarse geometry of Hartnell's firefighter problem on infinite
  graphs}.
\newblock {\em Discrete Mathematics}, 340(5):935--950, 2017.

\bibitem{3/2firefighter}
Ohad~N. Feldheim and Rani Hod.
\newblock 3/2 firefighters are not enough.
\newblock {\em Discrete Applied Mathematics}, 161(1-2):301--306, 2013.

\bibitem{Finbow.king}
Stephen Finbow, Andrew King, Gary Macgillivray, and Romeo Rizzi.
\newblock The firefighter problem for graphs of maximum degree three.
\newblock {\em Discrete Mathematics}, 307(16):2094--2105, 2007.

\bibitem{Finbow.MacGillivray}
Stephen Finbow and Gary MacGillivray.
\newblock The firefighter problem: a survey of results, directions and
  questions.
\newblock {\em Australasian Journal of Combinatorics}, 43(6):57--77, 2009.

\bibitem{fogarty}
Patricia Fogarty.
\newblock {\em Catching the Fire on Grids, PhD thesis}.
\newblock University of Vermont, 2003.

\bibitem{Fomin.fedor}
Fedor~V. Fomin, Pinar Heggernes, and Erik~Jan van Leeuwen.
\newblock The firefighter problem on graph classes.
\newblock {\em Theoretical Compututer Science}, 613(C):38--50, February 2016.

\bibitem{Hartke.Stephen}
Stephen~G. Hartke.
\newblock Attempting to narrow the integrality gap for the firefighter problem
  on trees.
\newblock In {\em Discrete Methods in Epidemiology}, pages 225--232, 2004.

\bibitem{Hartnell.Li}
Bert Hartnell and Qiyan Li.
\newblock Firefighting on trees: How bad is the greedy algorithm?
\newblock {\em Congressus Numerantium}, pages 187--192, 2000.

\bibitem{hartnell}
Berth Hartnell.
\newblock {Firefighter! An application of domination.}
\newblock 1995.
\newblock presented at the 10th Conference on Numerical Mathematics and
  Computing, University of Manitoba in Winnipeg, Canada.

\bibitem{improved-treedeg3}
Yutaka Iwaikawa, Naoyuki Kamiyama, and Tomomi Matsui.
\newblock Improved approximation algorithms for firefighter problem on trees.
\newblock {\em IEICE Transactions on Information and Systems},
  E94.D(2):196--199, 2011.

\bibitem{Lehner}
Florian Lehner.
\newblock {Firefighting on trees and Cayley graphs}.
\newblock {\em ArXiv e-prints, (arXiv:1707.01224v1 [math.CO])}, July 2017.

\bibitem{lyons_peres_2017}
Russell Lyons and Yuval Peres.
\newblock {\em Probability on Trees and Networks}.
\newblock Cambridge Series in Statistical and Probabilistic Mathematics.
  Cambridge University Press, 2017.

\bibitem{MacGillivray.wang}
Gary MacGillivray and Ping Wang.
\newblock On the firefighter problem.
\newblock {\em Journal of Combinatorial Mathematics and Combinatorial
  Computing}, 47:83--96, 2003.

\bibitem{Messinger08}
Margaret-Ellen Messinger.
\newblock Average firefighting on infinite grids.
\newblock {\em The Australasian Journal of Combinatorics}, 41:15--28, 2008.

\bibitem{wang-moeller}
Ping Wang and Stephanie~A. Moeller.
\newblock Fire control on graphs.
\newblock {\em Journal of Combinatorial Mathematics and Combinatorial
  Computing}, 41:19--34, 2002.

\end{thebibliography}

\end{document}